\documentclass[12pt,draftclsnofoot,onecolumn]{IEEEtran}
\usepackage{epsfig,latexsym,amsmath,amssymb,setspace,color,graphicx,algorithm,algorithmic,cases}
\pdfoutput=1
\usepackage{amssymb,amsmath,amsfonts,bm,epsfig,graphicx,latexsym}
\usepackage{rotating,setspace,latexsym,epsf,color}
\usepackage{cite,authblk,epstopdf,footnote}
\usepackage{float}
\usepackage{tabularx}
\usepackage{array}
\usepackage[subrefformat=parens,labelformat=parens]{subfig}
\usepackage{bbm}
\usepackage{amsthm}
\usepackage{thmtools}
\usepackage{algorithm,algorithmic}

\usepackage{eqparbox}

\declaretheorem[style=plain,qed=$\blacksquare$]{theorem}
\declaretheorem[style=plain,name=Definition,qed=$\blacksquare$]{Definition}

\declaretheorem[style=plain,name=Remark,qed=$\blacksquare$]{remark}

\declaretheorem[style=plain,name=Proposition,qed=$\blacksquare$]{proposition}
\declaretheorem[style=plain,name=Corollary,qed=$\blacksquare$]{corollary}

\allowdisplaybreaks
\IEEEoverridecommandlockouts
\def\mc{\ensuremath\mathcal}
\begin{document}
\doublespacing
\title{Combination Networks with or without Secrecy Constraints: The Impact of Caching Relays\thanks{This work was supported in part by the National Science Foundation Grants CNS 13-14719 and CCF 17-49665. Sections I-III of this paper was presented in part at the IEEE International Symposium of Information Theory (ISIT) 2017.}}
\author{Ahmed A. Zewail}
\author{Aylin Yener}
\vspace{-.1 in}
\affil{\normalsize Wireless Communications and Networking Laboratory (WCAN)\\
The School of Electrical Engineering and Computer Science\\
The Pennsylvania State University, University Park, PA 16802.\\
\quad zewail@psu.edu   \qquad yener@engr.psu.edu}
\maketitle
\vspace{-0.9 in}
\begin{abstract}
\vspace{-.2 in}
This paper considers a two-hop network architecture known as a combination network, where a layer of relay nodes connects a server to a set of end users. In particular, a new model is investigated where the intermediate relays employ caches in addition to the end users. First, a new centralized coded caching scheme is developed that utilizes maximum distance separable (MDS) coding, jointly optimizes cache placement and delivery phase, and enables decomposing the combination network into a set virtual multicast sub-networks. It is shown that if the sum of the memory of an end user and its connected relay nodes is sufficient to store the database, then the server can disengage in the delivery phase and all the end users' requests can be satisfied by the caches in the network. Lower bounds on the normalized delivery load using genie-aided cut-set arguments are presented along with second hop optimality. Next recognizing the information security concerns of coded caching, this new model is studied under three different secrecy settings: 1) secure delivery where we require an external entity must not gain any information about the database files by observing the transmitted signals over the network links, 2) secure caching, where we impose the constraint that end users must not be able to obtain any information about files that they did not request, and 3) both secure delivery and secure caching, simultaneously. We demonstrate how network topology affects the system performance under these secrecy requirements. Finally, we provide numerical results demonstrating the system performance in each of the settings considered. 
\end{abstract}
\vspace{-.3 in}
\begin{IEEEkeywords}
\vspace{-.2 in}
Combination networks with caching relays, coded caching, maximum distance separable (MDS) codes, secure delivery, secure caching. 
 \end{IEEEkeywords}
\section{Introduction}\label{intro}
Caching is foreseen as a promising avenue to provide content based delivery services for $5$G systems and beyond \cite{almeroth1996use,korupolu2001placement}. Caching enables shifting the network load from peak to off-peak hours leading to a significant improvement in overall network performance. During off-peak hours, in the \textit{cache placement phase}, the network is likely to have a considerable amount of under-utilized wireless bandwidth which is exploited to place \textit{functions} of data contents in the cache memories of the network nodes. This phase takes place prior to the end users' content requests, and thus content needs to be placed in the caches without knowing what specific content each user will request. The cached contents help reduce the required transmission load when the end users actually request the contents, during peak traffic time, known as the \textit{delivery phase}, not only by alleviating the need to download the entire requested data, but also by facilitating multicast transmissions that benefit multiple end users\cite{maddah2014fundamental}. As long as the storage capabilities increase, the required transmission load during peak traffic can be decreased, leading to the rate-memory trade-off \cite{maddah2014fundamental,ji2014fundamental}.

Various network topologies with caching capabilities have been investigated to date, see for example  \cite{karamchandani2014hierarchical,maddah2015cache,shariatpanahi2015multi,
ji2015caching,ji2015fundamental,tang2016coded,wan2017combination,wan2017novel,wan2017caching}. References \cite{karamchandani2014hierarchical,
ji2015caching,ji2015fundamental,tang2016coded,wan2017combination} have studied two-hop cache-aided networks. Reference \cite{karamchandani2014hierarchical} has studied hierarchical networks, where the server is connected to a set of relay nodes via a shared multicast link and the end users are divided into equal-size groups such that each group is connected to only one relay node via a multicast link. Thus, one relay needs to be shared by multiple users. We will not consider this model. 
 
 A fundamentally different model is investigated in references \cite{ji2015caching} and \cite{ji2015fundamental} where multiple overlapping relays serve each user. In this symmetric layered network, known as a \textit{combination network} \cite{ngai2004network}, the server is connected to a set of $h$ relay nodes, and each end user is connected to exactly $r$ relay nodes, thus each relay serves ${h-1}\choose {r-1}$ end nodes. In these references, end users randomly cache a fraction of bits from each file subject to the memory capacity constraint. Two delivery strategies have been proposed: one relies on routing the requested bits via the network links and the other is based on coded multicasting and combination network coding techniques \cite{xiao2007binary}. More recently, reference \cite{tang2016coded} has considered a class of networks which satisfies the resolvability property, which includes combination networks where $r$ divides $h$ \cite{baranyai1975factorization}. A centralized coded caching scheme has been proposed and shown to outperform, analytically and numerically, those in \cite{ji2015caching} and \cite{ji2015fundamental}. The cache allocation of \cite{tang2016coded} explicitly utilizes resolvability property, so that one can design the cache contents that make each relay node see the same set of cache allocations. In all of these references studying combination networks -resolvable or not-, only the end users are equipped with cache memories. 


In this paper, we boost the caching capabilities of combination networks by introducing caches at the relay nodes. In particular, we consider a general combination network equipped with caches at \textit{both the relay nodes and the end users}. The model in effect enables cooperation between caches from different layers to aid the server. We develop a new centralized coded caching scheme, by utilizing maximum distance separable (MDS) codes \cite{lin2004error} and jointly optimizing the cache placement and delivery phases. This proposed construction enables \textit{decomposing} the coded caching in combination networks into sub-problems in the form of the classical setup studied in \cite{maddah2014fundamental}. We show that if the sum of the memory size of a user and its connected relay nodes is large enough to store the library, then the server can disengage during the delivery phase altogether and all users' requests can be satisfied utilizing the cache memories of the relay nodes and end users. Genie-aided cut-set lower bounds on the transmission rates are provided. Additionally, for the special case, where there are no caches at the relays, we show that our scheme achieves the same performance of the scheme in \cite{tang2016coded} without requiring resolvability.  

In many practical scenarios, reliability is not the only consideration. Confidentiality, especially in file sharing systems, is also of paramount importance. Thus in the latter part of the paper, for the same model, we address the all important concerns of information security. Specifically, we consider combination networks with caches at the relays and end users, under three different scenarios. In the first scenario, we consider that the database files must be kept secret from any external eavesdropper that overhears the delivery phase, i.e., \textit{secure delivery} \cite{sengupta2015fundamental} \cite{awan2015fundamental}. In the second scenario, we consider that each user must only be able to decode its requested file and should not be able gain any information about the contents of the remaining files, i.e., \textit{secure caching} \cite{ravindrakumar2016fundamental} \cite{zewail2016fundamental}. Last, we consider both secure delivery and secure caching, simultaneously. We note that, in security for cache-aided combination networks, the only previous work consists of our recent effort \cite{zewail2016coded}, where the schemes are limited to resolvable combination networks with no caching relays.

 For all the considered scenarios, our proposed schemes based on the decomposition turn out to be optimal with respect to the total transmission load per relay, i.e., we achieve the cut set bound. Our study demonstrates the impact of cache memories at the relay nodes (in addition to the end users) in reducing the transmission load of the server. In effect, these caches can cooperatively replace the server during the delivery phase under sufficient total memory. Furthermore, we demonstrate the impact of the network topology on the system performance under secrecy requirements. In particular, we demonstrate that satisfying the \textit{secure caching} requirement does not require encryption keys and is feasible even with memory size less than the file size, unlike the case in references \cite{ravindrakumar2016fundamental} and \cite{zewail2016fundamental}. In addition, we observe that the cost due the \textit{secure delivery} is almost negligible in combination networks, similar to the cases in references \cite{sengupta2015fundamental} and \cite{awan2015fundamental} for other network topologies.

The remainder of the paper is organized as follows. Section \ref{sec:sm} describes the system model. In Section \ref{sec:nonsec}, we propose a new centralized coded caching scheme that is applicable to any cache-aided combination network. In Sections \ref{sec:achdelivery}, \ref{sec:achconf} and \ref{sec:achconf_sec}, we detail the achievability techniques for the three secrecy scenarios. In Section \ref{sec:discuss}, we provide the numerical results and discuss the insights learned from them. Section \ref{sec:con} concludes the paper. 
\vspace{-.09 in} 
\section{System Model}\label{sec:sm}
\vspace{-.03 in} 
\subsection{Network Model}
Consider a combination network, where the server, $S$, is connected to $K$ end users via a set of $h$ relay nodes. More specifically, each end user is connected to a distinct set of $r$ relay nodes, $r<h$, with $K={h \choose r}$. Each relay node is connected to $\hat{K}={{h-1}\choose{r-1}}=\frac{rK}{h}$ end users. Similar to references \cite{ji2015caching,ji2015fundamental,tang2016coded}, all network links are \textit{unicast}. In addition, similar to references  \cite{maddah2014fundamental,ji2014fundamental,sengupta2015fundamental,awan2015fundamental,
ravindrakumar2016fundamental,zewail2016fundamental,karamchandani2014hierarchical,ji2015caching,ji2015fundamental,tang2016coded}, all network links are assumed to be noiseless. Let $\mathcal{R}\!=\!\{\Gamma_1,..,\Gamma_h\}$ denote the set of relay nodes, and $\mathcal{U}\!=\!\{U_1,..,U_K\}$ the set of all end users. We denote the set of end users connected to $\Gamma_i$ by $\mathcal{N}(\Gamma_i)$, $|\mathcal{N}(\Gamma_i)|\!=\!\hat{K}$ for $i=1,..,h$, and the set relay nodes connected to user $i$ by $\mathcal{N}(U_i)$, $|\mathcal{N}(U_i)|=r$.
\begin{figure}[t]
\includegraphics[width=4.3 in,height=2.7 in]{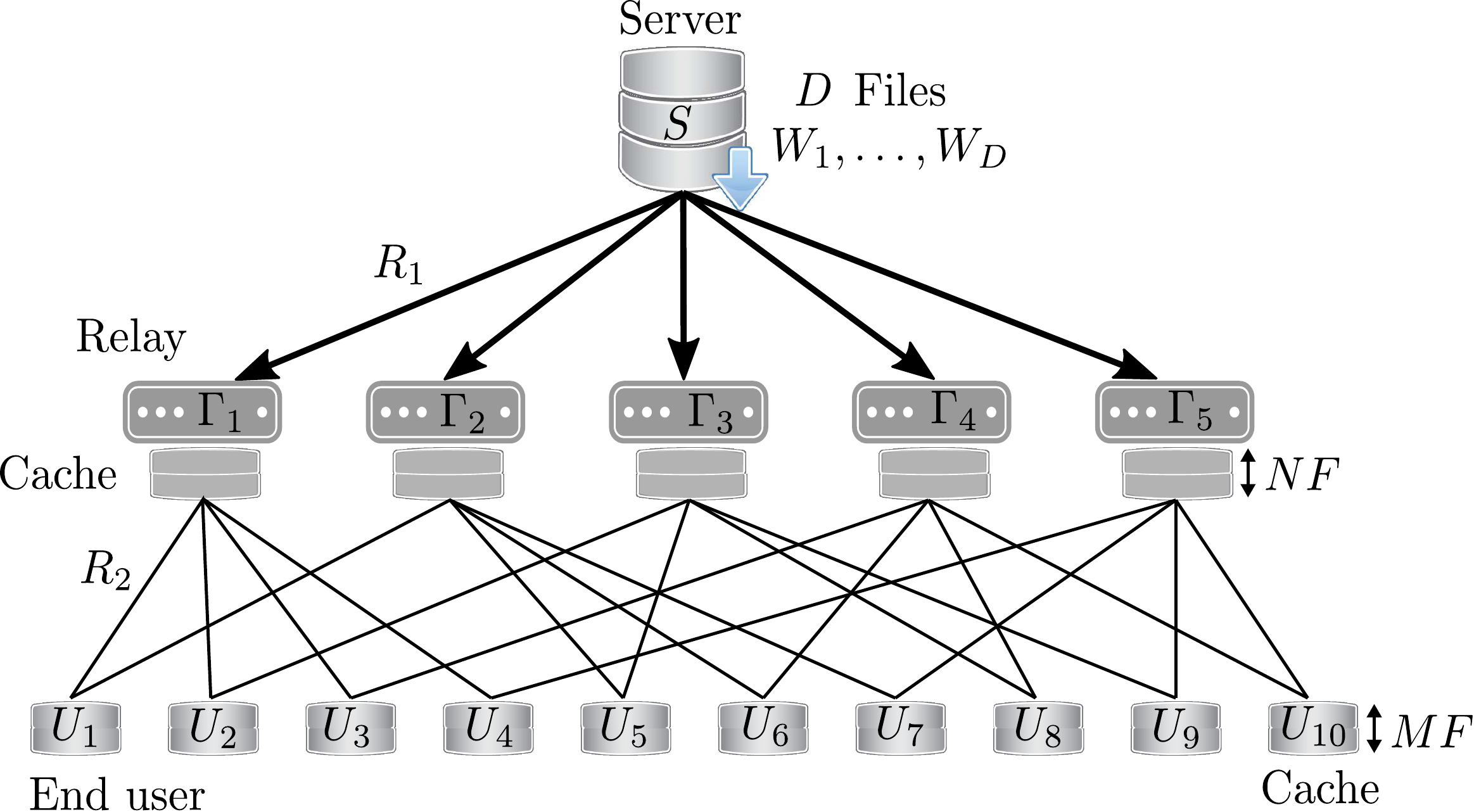}
\centering
\caption{\small A combination network with $K\!=\!10$, $h\!=\!5$, $r\!=\!2$, and caches at both relays and end users.}\label{Fig:sym}
\centering
\vspace{-.1 in}
\end{figure}
%
The function $Index(,): (i,k)\rightarrow \{1,..,\hat{K}\}$, where $i\in\{1,..,h\}$ and $k\in \mathcal{N}(\Gamma_i)$, is defined as a function that orders the end users connected to relay node $\Gamma_i$ in  an ascending manner. For example, for the network in Fig. \ref{Fig:sym}, $\mathcal{N}(\Gamma_2)=\{1,5,6,7\}$, $\mathcal{N}(\Gamma_4)=\{3,6,8,10\}$, and
\vspace{-.1 in}
\begin{align}\nonumber
&Index(2,1)=1, \ Index(2,5)=2, \ Index(2,6)=3, \ Index(2,7)=4, \nonumber \\
&Index(4,3)=1, \ Index(4,6)=2, \ Index(4,8)=3, \ Index(4,10)=4. \nonumber 
\end{align}
\vspace{-.3 in}
\subsection{Caching Model}
Server $S$ has $D$ files, $W_1,.., W_D$, each with size $F$ bits. We treat the case where the number of users is less than or equal to the number of files, i.e., $K\leq D$. Each end user is equipped with a cache memory of size $MF$ bits while each relay node has cache memory of size $NF$ bits, i.e., $M$ and $N$ denote the normalized cache memory sizes at the end users and relay nodes, respectively. The system operates in two phases.
\subsubsection{Cache Placement Phase} 
In this phase, the server allocates functions of its database files in the relay nodes end users caches. The allocation is done ahead of and without the knowledge of the demand of the individual users. 
\begin{Definition}
(Cache Placement): 
The content of the cache memories at relay node $j$ and user $k$, respectively are given by 
\begin{equation}
V_{j} = \nu_{j} (W_1, W_2,..,W_D), \qquad  Z_{k} = \phi_{k} (W_1, W_2,..,W_D),
\end{equation}
where $\nu_{j} : [2^F]^D\rightarrow [2^F]^N$ and $\phi_{k} : [2^F]^D\rightarrow [2^F]^M$, i.e., $H(V_{j})\leq NF$ and $H(Z_{k})\leq MF.$ 
\end{Definition}
\subsubsection{Delivery Phase}
Each user requests a file independently and randomly \cite{maddah2014fundamental}. Let $d_k$ denote the index of the requested file by user $k$, i.e., $d_k\in\{1,2,..,D\}$; $\bm d$ represents the demand vector of all users. The server responds to users' requests by transmitting signals to the relay nodes. Then, each relay transmits \textit{unicast} signals to its connected end users.   
From the $r$ received signals and $Z_k$, user $k$ must be able to reconstruct its requested file $W_{d_k}$.  
\begin{Definition}
(Coded Delivery): The mapping from the database files, $\{W_1,..,W_D\}$, and the demand vector $\bm d$ into the transmitted signal by the server to $\Gamma_i$ is given by the encoding function
\begin{equation}
X_{i,\bm d} = \psi_i(W_1,..,W_D,\bm d), \qquad i =1,2,..,h, 
\end{equation}
where $\psi_i: [2^F]^D\times \{1,..,D\}^K\rightarrow [2^F]^{R_1}$, and $R_1$ is the rate, normalized by the file size, $F$, of the transmitted signal from the server to each relay node. The transmitted signal from $\Gamma_i$ to user $k\in \mathcal{N}(\Gamma_i)$, is given by the encoding function 
\begin{equation}
Y_{i,\bm d,k} = \varphi_{k} (X_{i,\bm d},V_i,\bm d),  
\end{equation}
where $\varphi_{k} : [2^F]^{R_1}\times[2^F]^{N}\times \{1,..,D\}^K\rightarrow [2^F]^{R_2}$, and $R_2$ is the normalized rate of the transmitted signal from a relay node to a connected end user.
User $k$ recovers its requested file by
\begin{equation}
\hat W_{k} = \mu_{k}(Z_k,\bm d,\{Y_{i,\bm d,k}:i\in \mathcal{N}(U_k)\}),  
\end{equation}   
where $\mu_k: [2^F]^M\times \{1,..,D\}^K\times [2^F]^{rR_2}\rightarrow [2^F]$ is the decoding function.
\end{Definition}
We require that each end user $k$ recover its requested file reliably, i.e., for any $\epsilon>0$, 
\begin{equation}\label{reliabeconst}
\max_{\bm d,k} P(\hat W_{d_k}\neq W_{d_k})<\epsilon.
\end{equation}
For a positive integer, $L$, we will use the notation $[L] \triangleq \{1,.., L\}$.
\section{A new coded caching scheme for combination networks}\label{sec:nonsec}
 We develop a new caching scheme for general cache-aided combination networks. In addition, we show that the upper bound derived in \cite{tang2016coded} for resolvable combination networks, is in fact achievable for all combination networks. 

The main idea behind our proposed scheme is that each file is encoded using an $(h,r)$ maximum distance separable (MDS) code \cite{lin2004error,yeung2008information}. Then, each relay node acts as a server for one of the resulting encoded symbols. Since each end user is connected to $r$ different relay nodes, by the end of the delivery phase, it will be able to obtain $r$ different encoded symbols that can be used to recover its requested file. 
\vspace{-.21 in}
\subsection{Cache Placement Phase}\label{sec:nonsecplac}
As a first step, the server divides each file into $r$ equal-size subfiles. Then, it encodes them using an $(h,r)$ maximum distance separable (MDS) code \cite{lin2004error,yeung2008information}. We denote by $f_n^i$ the resulting encoded symbol, where $n$ is the file index and $i=1,2,..,h$. The size of each encoded symbol, $f_n^i$, is $F/r$ bits, and any $r$ encoded symbols are sufficient to reconstruct the file $n$. The server divides each encoded symbol into two parts, $f_n^{i,1}$ and $f_n^{i,2}$, such that the size of $f_n^{i,1}$ is $\frac{NF}{D}$ bits, and the size of $f_n^{i,2}$ is $(\frac{1}{r}-\frac{N}{D})F$ bits.

We describe the achievability for $M=\frac{(t_1-t_2)Nr}{\hat{K}}+\frac{t_2D}{\hat{K}}$, $t_1 \!\in\!\{0,1,..,\min(\hat{K},\lfloor\frac{\hat{K}N}{D}\rfloor\}$, and $ t_2\!\in\!\{0,1,..,\hat{K}\}$, noting that the convex envelope is achievable by memory sharing \cite{maddah2014fundamental}. 
First, the server places $f_{n}^{j,1}$, $\forall n$ in the cache memory of relay node $\Gamma_j$. Then, user $k$, with $k \in \mathcal{N}(\Gamma_j)$, caches a random fraction of $\frac{t_1}{\hat{K}}$ bits from $f_{n}^{j,1}$, $\forall n$, which we denote by $f_{n,k}^{j,1}$. On the other hand, $f_n^{i,2}$ is divided into ${{\hat{K}}\choose {t_2}}$ disjoint pieces each of which is denoted by $f_{n,\mathcal{T}}^{i,2}$, where $n$ is the file index, i.e., $n\in[D]$, $i$ is the index of the encoded symbol, $i=1,..,h$, and $\mathcal{T}\!\subseteq[\hat{K}], |\mathcal{T}|\!=\!t_2$. The size of each piece is $\frac{(\frac{1}{r}-\frac{N}{D})}{{{\hat{K}}\choose {t_2}}}F$ bits. 
 The server allocates the pieces $f_{n,\mathcal{T}}^{j,2}$, $\forall n$ in the cache memory of user $k$ if $k \in \mathcal{N}(\Gamma_j)$ and $Index(j,k)\in \mathcal{T}$. Therefore, the cache contents at the relay nodes and end users are given by 
 \begin{equation}
V_k\!=\!\left\{f_{n}^{k,1}:  \ \forall n \right\},\\  \ \qquad Z_k\!\!=\left\{ f_{n,k}^{j,1},f_{n,\mathcal{T}}^{j,2}\!\!:\! j\!\in\! \mathcal{N}(U_k),\! \ \!  Index(j,k)\!\in\! \mathcal{T},  \forall n\right\}.
 \end{equation}
Clearly, this satisfies the memory constraint at each relay node. 
 The number of the accumulated bits at the cache memory of each end user is given by 
\begin{equation}
Dr \frac{N}{D}\frac{t_1}{\hat{K}}F +Dr \frac{(\frac{1}{r}-\frac{N}{D})}{{{\hat{K}}\choose {t_2}}}F {{\hat{K}\!-1}\choose {t_2\!-1}}=\frac{Nt_1r}{\hat{K}}F+\frac{(D-Nr)t_2}{\hat{K}}F =MF,
\end{equation} 
which satisfies the memory constraint. We summarize the cache placement procedure in Algorithm \ref{alg_place1}.
\begin{algorithm}[t]
\begin{algorithmic}[1]
\REQUIRE $\{ W_{1}, \dots, W_{D}\}$ 
\ENSURE $ Z_k, k \in [K]$
\FOR{$l \in [D]$}
\STATE Encode each file using an $(h,r)$ MDS code $\rightarrow f_l^i$, $i=1,..,h$.
\FOR{$i\in [h]$}
\STATE Divide $f_l^{i}$ into $f_l^{i,1}$ with size $\frac{NF}{D}$ bits and $f_l^{i,2}$ with size $(\frac{1}{r}-\frac{N}{D})F$ bits.
\STATE $V_i \leftarrow  f_l^{i,1}$
\STATE Partition $f_l^{i,2}$ into equal-size pieces $f_{l,\mathcal{T}}^{i,2}$, $\mathcal{T}\!\subseteq\![\hat{K}]$ and $|\mathcal{T}|\!=\!t_2$. 
\ENDFOR 
\ENDFOR

\FOR{$k \in [K]$}
\STATE User $k$ caches a random fraction $\frac{t_1}{\hat{K}}$ bits from $f_{n}^{j,1}$, $\forall n$ $\rightarrow f_{n,k}^{j,1}$
\STATE $Z_k \leftarrow \bigcup_{j \in \mathcal{N}(U_k)}\bigcup_{l \in [N]} \left\{f_{l,\mathcal{T}}^j:  Index(j,k)\in \mathcal{T} \right\}\bigcup{f_{l,k}^{j,1}}$
\ENDFOR
\end{algorithmic}
 \caption{Cache placement procedure}\label{alg_place1}
\end{algorithm}
\vspace{-.2 in}
\subsection{Coded Delivery Phase}\label{sec:nonsecdel}
After announcing the demand vector to the network, the server and the relays start to serve the end users' requests. For each relay $\Gamma_j$, at each transmission instance, we consider $\mathcal{S}\subseteq[\hat{K}]$, where $|\mathcal{S}|=t_2+1$. For each choice of $\mathcal{S}$, the server transmits to the relay node $\Gamma_j$, the signal 
\begin{equation}
X_{j,\bm d}^{\mathcal{S}}=\bigoplus_{\{i:i \in \mathcal{N}(\Gamma_j),\ Index(j,i)\in \mathcal{S}\}} f^{j,2}_{d_i,\mathcal{S} \setminus\{Index(j,i)\}}.
\end{equation}
In total, the server transmits to $\Gamma_j$ the signal
\begin{equation}
X_{j,\bm d}=\bigcup_{\mathcal{S}\subseteq[\hat K]:|\mathcal{S}|=t_2+1} \left\{X_{j,\bm d}^{\mathcal{S}}\right\}.
\end{equation}  
 $\Gamma_j$ forwards the signal $X_{j,\bm d}^{\mathcal{S}}$ to user $i$ whenever $Index(j,i)\in\mathcal{S}$. In addition, $\Gamma_j$ transmits the missing bits from $f_{d_i}^{j,1}$ to user $i$, $i\in \mathcal{N}(\Gamma_j)$. The transmitted signal from $\Gamma_j$ to user $i$ is 
\begin{align}
Y_{j,\bm d,i}=\Big\{ f_{d_i}^{j,1} \setminus f_{d_i,i}^{j,1}\Big\} \bigcup_{\mathcal{S}\subseteq[\hat K]:|\mathcal{S}|=t_2+1,Index(j,i)\in\mathcal{S}} \big\{X_{j,\bm d}^{\mathcal{S}}\big\}.
\end{align} 
 User $i$ can recover $\left\{f^{j,2}_{d_i,\mathcal{T}}: \mathcal{T}\subseteq
 [\hat{K}]\setminus \{Index(j,i)\} \right\}$ from the signals received from $\Gamma_j$, utilizing its cache's contents. XORing these pieces to the ones already in its cache, i.e., $f^{j,2}_{d_i,\mathcal{T}}$ with $Index(j,i)\in \mathcal{T}$, user $i$ can recover the encoded symbol $f^{j,2}_{d_i}$. Additionally, from its received signal, user $i$ directly gets $f_{d_i}^{j,1}$. Therefore, it can obtain $f^{j}_{d_i}$. Since, user $i$ receives signals from $r$ different relay nodes, it can obtain the encoded symbols $f^{j}_{d_i}$, $\forall j \in \mathcal{N}(U_i)$, and is able to successfully reconstruct its requested file $W_{d_i}$. The delivery procedure is summarized in Algorithm \ref{alg_delv1}.
 \begin{algorithm}[t]
\begin{algorithmic}[1]
\REQUIRE $\bm d$
\ENSURE $X_{j,\bm d}, Y_{j,\bm d,i}, j\in [h], i\in [K] $
\FOR{$j\in [h]$}
\FOR{$\mc S  \in [\hat{K}], |\mc S|=t_2+1$} 
\STATE $X_{j,\bm d}^{\mathcal{S}}\leftarrow \bigoplus_{\{i:i \in \mathcal{N}(\Gamma_j),\ Index(j,i)\in \mathcal{S}\}} f^{j,2}_{d_i,\mathcal{S} \setminus\{Index(j,i)\}}$
\ENDFOR
\STATE $X_{j,\bm d}\leftarrow \bigcup_{\mathcal{S}\subseteq[\hat K]} \{X_{j,\bm d}^{\mathcal{S}}\}$
\FOR{$i\in \mc N(\Gamma_j)$}
\STATE $Y_{j,\bm d,i}\leftarrow \left\{ f_{d_i}^{j,1} \setminus f_{d_i,i}^{j,1}\right\}\bigcup_{\mathcal{S}\subseteq [\hat K]:Index(j,i)\in\mathcal{S}} \{X_{j,\bm d}^{\mathcal{S}}\}$
\ENDFOR
\ENDFOR
\end{algorithmic}
 \caption{Delivery procedure}\label{alg_delv1}
\end{algorithm} 
\vspace{-.2 in}
\subsection{Rates}
First, observe that the server transmits ${{\hat{K}}\choose {t_2+1}}$ sub-signals to each relay node, each of which has length $\frac{(\frac{1}{r}-\frac{N}{D})}{{{\hat{K}}\choose {t_2}}}F$ bits, thus the transmission rate in bits from the server to each relay node is 
\begin{align}
R_1F={{\hat{K}}\choose {t_2\!+\!1}} \frac{(\frac{1}{r}\!-\!\frac{N}{D})}{{{\hat{K}}\choose {t_2}}}F=\frac{(\hat{K}\!-\!t_2)(\frac{1}{r}\!-\!\frac{N}{D})}{t_2+1}F.
\end{align} 

During the second hop, each relay node forwards ${{\hat{K}-1}\choose {t_2}}$ from its received sub-signals to each of its connected end users. Additionally, it sends $(1-\frac{t_1}{\hat K})\frac{N}{D}F$ bits, from its cache memory to each of its connected end users. Therefore, we have
\begin{align}
R_2F&={{\hat{K}-1}\choose {t_2}} \frac{(\frac{1}{r}\!-\!\frac{N}{D})}{{{\hat{K}}\choose {t_2}}}F+(1-\frac{t_1}{\hat{K}})\frac{N}{D}F=\frac{(\hat{K}-t_2)(\frac{1}{r}\!-\!\frac{N}{D})}{\hat{K}}F+\frac{(\hat{K}-t_1)N}{D\hat{K}}F\nonumber \\
&=\frac{1}{r}\Big(1-\!\frac{t_2}{\hat{K}}\!-\!\frac{(t_1-t_2)Nr}{D\hat{K}}\Big)F=\frac{1}{r}\Big(1\!-\!\frac{M}{D}\Big)F.
\end{align} 
These findings are presented in the following theorem. 
 \begin{theorem}\label{thm_relay_cache_}
The normalized transmission rates, for $0 \leq N \leq \frac{D}{r}$, $M=\frac{(t_1-t_2)Nr}{\hat{K}}+\frac{t_2D}{\hat{K}}$, $t_1 \!\in\!\{0,1,..,\min(\hat{K},\lfloor\frac{\hat{K}N}{D}\rfloor\}$, and $ t_2\!\in\!\{0,1,..,\hat{K}\}$, are upper bounded by 
\begin{align}\label{nosecraterelay}
R_1\leq \frac{\hat{K}-t_2}{r(t_2+1)}\left(1-\frac{Nr}{D}\right), \quad R_2\leq\frac{1}{r}\left(1-\frac{M}{D}\right).
\end{align}
Furthermore, the convex envelope of these points is achievable.  
\end{theorem}
If $M$ is not in the form of $M=\frac{(t_1-t_2)Nr}{\hat{K}}+\frac{t_2D}{\hat{K}}$, we use memory sharing as in \cite{maddah2014fundamental,karamchandani2014hierarchical}. 

\vspace{-.05 in}
\begin{remark}
Observe that the caches at the relays help decrease the transmission load only during the first hop, $R_1$. The transmission load over the second hop, $R_2$, depends only on the size of end users' cache memories, $M$, as it is always equal to the complement of the local caching gain divided by the number of relay nodes connected to each end users. 
\end{remark}
\vspace{-.2 in}
\begin{remark}
It can be seen from (\ref{nosecraterelay}) that when $t_2=\hat{K}$, i.e., $M\geq D-Nr$, we can achieve $R_1=0$. In other words, whenever $M+Nr\geq D$, i.e., the total memory at each end user and its connected relay nodes is sufficient to store the whole file library, the server is not required to transmit during the delivery phase.
\end{remark}
\vspace{-.1 in}
When there are no caches at the relays \cite{ji2015caching,ji2015fundamental}, i.e., setting $N=0$ and $t_1=0$, we obtain the following result.
\vspace{-.1 in}
\begin{corollary}
The normalized transmission rates, for $N=0$, $M\!=\frac{tD}{\hat{K}}$, and $t\!\in\!\{0,1,..,\hat{K}\}$, are upper bounded by 
\begin{align}\label{nosecrate}
R_1\leq\frac{\hat{K}}{r}\left(1-\dfrac{M}{D}\right)\frac{1}{1+\frac{\hat{K}M}{D}}, \ \qquad R_2\leq\frac{1}{r}\left(1-\frac{M}{D}\right).
\end{align}
In addition, the convex envelope of these points is achievable.  
\end{corollary}

\begin{remark}\label{remark_resolvable1}
The achievable rates in (\ref{nosecrate}) are the same as the ones in \cite{tang2016coded} which have been shown to be achievable for a special class of combination networks where $r$ divides $h$, i.e., resolvable networks. By our scheme, we have just demonstrated that the resolvability property is not necessary to achieve these rates. Furthermore, it has been shown in \cite{tang2016coded} that, for resolvable networks, these rates outperform the ones in \cite{ji2015caching} and \cite{ji2015fundamental}. Thus, our proposed scheme outperforms the ones in \cite{ji2015caching} and \cite{ji2015fundamental}.
\end{remark}
\begin{remark}
One can see from (\ref{nosecrate}) that the upper bound on $R_1$ is formed by the product of three terms. The first term $\frac{\hat{K}}{r}$ is due the fact that each relay node is connected to $\hat{K}$ end users, each of which is connected to $r$ relay nodes. Thus, each relay node is responsible for $\frac{1}{r}$ of the load on a server that is connected to $\hat{K}$ end users. The second term $(1-\frac{M}{D})$ represents the local caching gain at each end user. The term $\frac{1}{1+\frac{\hat{K}M}{D}}$ represents the global caching gain of the proposed scheme. 
\end{remark}
The merit our proposed scheme is that it allows us to virtually decompose the combination network into a set of sub-networks, each of which in the form of the multicast network \cite{maddah2014fundamental}. In particular, for the case where $N\!=\!0$, each relay node acts as a virtual server with library of $D$ files each of size $F/r$ bits, while each connected end user dedicates $1/r$ from its memory to this library. Therefore, any scheme developed for the classical multicast setup \cite{maddah2014fundamental} which achieves rate $R_{\text{Multicast}}(MF/r,D,\hat{K},F/r)$ can be utilized in the context of combination networks and achieves rate $R_1=R_{\text{Multicast}}$. In other words, for large enough $F$, schemes developed for the cases where the users' demands are non-uniform \cite{niesen2016coded}, the number of user is greater than the number of files, \cite{wan2016caching}, for small values of the end users memories \cite{chen2014fundamental}, utilizing coded prefetching \cite{tian2016caching}, can be adopted in a combination network after the decomposition step via MDS coding. 

In addition, by applying the proposed decomposition, we can utilize any scheme that is developed for combination networks with no relay caches, $N=0$, in the case where the relays are equipped with cache memories, i.e., $0<N\leq \frac{D}{r}$, as indicated in the following proposition.
\begin{proposition}
 Suppose that the rate pair $R^{N=0}_1(MF,D,K,F)$ and $R^{N=0}_2(MF,D,K,F)$ is known to be achievable in a combination network with no relay caches. Then, for a combination network with relay cache of size $NF$ bits, the rate pair $R_1=R^{N=0}_1(M_1F,D,K,$ $(1-\frac{Nr}{D})F)$ and $R_2=R^{N=0}_2(M_1F,D,K,(1-\frac{Nr}{D})F)+(Nr-M_2)\frac{F}{rD}$ is achievable for any choice of $M_1,M_2\geq 0$ and $M_1+M_2\leq M$.
\end{proposition}
\begin{proof}
Split each file of the database, $W_i$ into two subfiles $W_i^1$ of size $\frac{Nr}{D}F$ bits and $W_i^2$ of size $(1-\frac{Nr}{D})F$ bits. Encode each of subfiles $\{W_i^1, \ \forall i\}$ using an $(h,r)$ MDS code. Each encoded symbol is cached by one of the relays. Divide the cache of each end user into two partitions of sizes $M_1F$ and $M_2F$ such that $M_1+M_2\leq M$. The partition of $M_1F$ bits is dedicated to the library formed by the subfiles $\{W_i^2, \ \forall i\}$, for which we apply any caching scheme that is known for a combination networks with no relay caches. The second partition of size $M_2F$ is filled by bits from the memories of relays connected to the end user as explained in subsection \ref{sec:nonsecplac}, leading to the achievable pair in the proposition.
\end{proof}
\vspace{-.3 in}
\subsection{An Illustrative Example}
We illustrate our proposed scheme by an example. Consider the network depicted in Fig. \ref{Fig:sym}, where $D\!=\!10$, $N=0$ and $M\!=\!\frac{15}{2}$, i.e., $t=3$. This network is not resolvable. 
\subsubsection{Cache Placement Phase}
Each file, $W_n$, is divided into $2$ subfiles. Then, the server encodes them using an $(5,2)$ MDS code. We denote the resulting encoded symbols by $f_n^j$, where $n$ is the file index, i.e., $n=1,..,10$, and $j=1,..,5$. Furthermore, we divide each encoded symbol into $4$ pieces each of size $\frac{F}{8}$ bits, and denoted by    
$f_{n,\mathcal{T}}^j$, where $\mathcal{T}\!\subseteq\! [4]$ and $|\mathcal{T}|\!=\!3$.
The contents of the cache memories at the end users are given in Table I. 
\begin{table}
\begin{center}
\small
 \begin{tabular}{||c| c ||} 
 \hline
 User $i$ & $Z_i$ \\ [0.5ex] 
 \hline\hline
 1 & $\left\{
f_{n,123}^1,f_{n,124}^1,f_{n,134}^1,f_{n,123}^2,f_{n,124}^2,f_{n,134}^2: \forall n  
\right\}$ \\ 
 \hline
 2 &$\left\{ 
f_{n,123}^1,f_{n,124}^1,f_{n,234}^1,f_{n,123}^3,f_{n,124}^3,f_{n,134}^3: \forall n 
\right\}$ \\
 \hline
 3 & $\left\{ 
f_{n,123}^1,f_{n,134}^1,f_{n,234}^1,f_{n,123}^4,f_{n,124}^4,f_{n,134}^4: \forall n 
\right\}$  \\
 \hline
 4 & $\left\{ 
f_{n,124}^1,f_{n,134}^1,f_{n,234}^1,f_{n,123}^5,f_{n,124}^5,f_{n,134}^5:  \forall n  
\right\}$  \\
 \hline
 5 & $\left\{
f_{n,123}^2,f_{n,124}^2,f_{n,234}^2,f_{n,123}^3,f_{n,124}^3,f_{n,234}^3:  \forall n  
\right\}$  \\
 \hline
 6 & $\left\{ 
f_{n,123}^2,f_{n,134}^2,f_{n,234}^2,f_{n,123}^4,f_{n,124}^4,f_{n,234}^4:  \forall n  
\right\}$  \\
 \hline
 7 & $\left\{
f_{n,124}^2,f_{n,134}^2,f_{n,234}^2,f_{n,123}^5,f_{n,124}^5,f_{n,234}^5:  \forall n  
\right\}$  \\
 \hline
 8 & $\left\{
f_{n,123}^3,f_{n,134}^3,f_{n,234}^3,f_{n,123}^4,f_{n,134}^4,f_{n,234}^4:  \forall n  
\right\}$  \\
 \hline
 9 & $\left\{
f_{n,124}^3,f_{n,134}^3,f_{n,234}^3,f_{n,123}^5,f_{n,134}^5,f_{n,234}^5:  \forall n  
\right\}$  \\
 \hline
 10 & $\left\{
f_{n,124}^4,f_{n,134}^4,f_{n,234}^4,f_{n,124}^5,f_{n,134}^5,f_{n,234}^5:  \forall n  
\right\}$ \\ 
 \hline
\end{tabular}
\normalsize
\end{center}
\caption{The cache contents at the end users for $N=K=10$ and $M=\frac{15}{2}$.}
\end{table}
Observe that each user stores $6$ pieces of the encoded symbols of each file, i.e., $\frac{3}{4}F$ bits, which satisfies the memory constraint.  
\subsubsection{Coded Delivery Phase}
Assume that user $k$ requests the file $W_k$, and $k=1,..,10$. The server transmits the following signals 
\begin{equation}\nonumber
X_{1,\bm d}=
f_{4,123}^1\oplus f_{3,124}^1\oplus f_{2,134}^1\oplus f_{1,234}^1  
, \qquad 
X_{2,\bm d}= 
f_{7,123}^2\oplus f_{6,124}^2\oplus f_{5,134}^2\oplus f_{1,234}^2  
,
\end{equation}
\begin{equation}\nonumber
X_{3,\bm d}= 
f_{9,123}^3\oplus f_{8,124}^3\oplus f_{5,134}^3\oplus f_{2,234}^3  
, \qquad
X_{4,\bm d}= 
f_{10,123}^4\oplus f_{8,124}^4\oplus f_{6,134}^4\oplus f_{4,234}^4  
,
\end{equation}
\begin{equation}\nonumber
\mbox{ and }X_{5,\bm d}=
f_{10,123}^5\oplus f_{9,124}^5\oplus f_{7,134}^5\oplus f_{4,234}^5  
.
\end{equation}
Then, each relay node forwards its received signal to the set of connected users, i.e.,
$ Y_{i,\bm d,k}=X_{i,\bm d}, \  \forall k\in \mathcal{N}(\Gamma_i)$.
 The size of each transmitted signal is equal to the size of a piece of the encoded symbols, i.e., $\frac{1}{8}F$. Thus, $R_1=R_2=\frac{1}{8}$. Now, utilizing its memory, user $1$ can extract the pieces $f_{1,234}^1$ and $f_{1,234}^2$ from the signals received from relay nodes $\Gamma_1$ and $\Gamma_2$, respectively. Therefore, user $1$ reconstructs $f_{1}^1$ and $f_{1}^2$, and decodes its requested file $W_1$. Similarly, user $2$ reconstructs $f_{2}^1$ and $f_{2}^3$, then decodes $W_2$, and so on for the remaining users. 
\subsection{Lower Bounds}
Next, we derive genie-aided lower bounds on the delivery load.
\subsubsection{Lower bound on $R_1$}\label{sec:lower_R1} 
Consider a cut that contains $l$ relay nodes, $l \in\{r,..,h\}$, and $s$ end users from the ${l}\choose{r}$ end users who are connected exclusively to these $l$ relay nodes. The remaining end users are served by a genie. Suppose at the first request instance, these $s$ users request the files $W_1$ to $W_s$. Then, at the second request instance, they request the files $W_{s+1}$ to $W_{2s}$, and so on till the request instance $\lfloor D/s \rfloor$. In order to satisfy all users' requests, the total transmission load from the server and the total memory inside the cut must satisfy 
\begin{equation}
H(W_1,..,W_{s\lfloor D/s \rfloor})=s\lfloor D/s \rfloor F \leq \lfloor D/s \rfloor lR_1F+sMF+lNF. 
\end{equation}
Therefore, we can get 
\begin{equation}
R_1 \geq \frac{1}{l}\!\left(s\!-\!\frac{sM+l N}{\lfloor D/s\rfloor}\right).
\end{equation}
Similar to \cite[Appendix B-A]{ji2015caching}, the smallest number of relay nodes serving a set of $x$ users equals to $u=\min(x+r-1,h)$. Therefore, by the cut set argument, we can get
 \begin{equation}
R_1 \geq   \frac{1}{u}\!\left(x\!-\!\frac{xM+u N}{\lfloor D/x\rfloor}\right). 
\end{equation}
\subsubsection{Lower bound on $R_2$}\label{sec:lower_R2} 
Consider the cut that contains user $k$ only. Assume $N$ request instances such that at instance $i$, user $k$ requests the file $W_i$. Then, we have the following constraint in order to satisfy the user's requests 
\begin{equation}
H(W_1,..,W_{D})=DF\leq D r R_2+MF. 
\end{equation}
Therefore, we can get the following bound on $R_2$
\begin{equation}
R_2\geq \frac{1}{r}\left(1-\frac{M}{D}\right).
\end{equation}
Now, taking into account all possible cuts, we have the following theorem. 
\begin{theorem}\label{thm:lower}
The normalized transmission rates for $0<M+rN\leq D$ are lower bounded by 
\begin{equation}
R_1 \geq \max \left(\max_{l\in\{r,..,h\}} \max_{s\in\! \{1,..,\min(D,{{l}\choose{r}})\}}\frac{1}{l}\!\left(s\!-\!\frac{sM+l N}{\lfloor D/s\rfloor}\right),\!\!\max_{x\in\{1,..,\min(D,K)\}}  \frac{1}{u}\!\left(x\!-\!\frac{xM+u N}{\lfloor D/x\rfloor}\right)\right),
\end{equation}
where $u=\min(x+r-1,h)$, and 
\begin{equation}
R_2\geq \frac{1}{r}\left(1-\frac{M}{D}\right).
\end{equation}
\end{theorem}
  In the following three sections, we investigate the cache-aided combination network under three different secrecy requirements.
\section{Coded caching with secure delivery}\label{sec:achdelivery}
First, we examine the system with \textit{secure delivery}. That is, 
we require that any external eavesdropper that observes the transmitted signals during the delivery phase, must not gain any information about the files, i.e., for any $\delta>0$
\begin{equation}\label{secureconst}
I(\mathcal{X},\mathcal{Y};W_1,..W_D)< \delta, 
\end{equation}
where $\mathcal{X}, \mathcal{Y}$ are the sets of transmitted signals by the server and the relay nodes, respectively. 

In order to satisfy (\ref{secureconst}), we place keys in the network caches during the placement phase. These keys are used to encrypt, i.e., one-time pad \cite{shannon1949}, the transmitted signals during the delivery phase as in \cite{sengupta2015fundamental} and \cite{awan2015fundamental}.  
\subsection{Cache Placement Phase}
We start by providing a scheme for $M=1+\frac{t_2(D-1)}{\hat{K}}+\frac{(t_1-t_2)r(D-1)N}{\hat{K}(N+\hat{K}-t_1)}$, where $t_1,t_2\!\in\!\{0,1,..,\hat{K}\}$ and $\frac{t_1}{\hat{K}}\leq\frac{N}{D+\hat{K}-t_1}$. Other values of $M$ are achievable by memory sharing. First, the server encodes each file using an MDS code to obtain the encoded symbols $\{f_n^i:i\in[h]\}$. Then, the servers divides each encoded symbol into two parts, $f_n^{i,1}$ with size $\frac{NF}{D+\hat{K}-t_1}$ bits and $f_n^{i,2}$ with size $\frac{F}{r}-\frac{NF}{D+\hat{K}-t_1}$ bits. Second, the server places $f_{n}^{j,1}$, $\forall n$ in the cache memory of relay node $\Gamma_j$. Then, user $k$, with $k \in \mathcal{N}(\Gamma_j)$, caches a random fraction of $\frac{t_1}{\hat{K}}$ bits from $f_{n}^{j,1}$, $\forall n$, which we denote by $f_{n,k}^{j,1}$. On the other hand, $f_n^{i,2}$ is divided into ${{\hat{K}}\choose {t_2}}$ disjoint pieces each of which is denoted by $f_{n,\mathcal{T}}^{i,2}$, where $n$ is the file index, i.e., $n\in[D]$, $i$ is the index of the encoded symbol, $i=1,..,h$, and $\mathcal{T}\!\subseteq[\hat{K}], |\mathcal{T}|\!=\!t_2$. The size of each piece is $\frac{\frac{F}{r}-\frac{NF}{D+\hat{K}-t_1}}{{{\hat{K}}\choose {t_2}}}F$ bits. 
 The server allocates the pieces $f_{n,\mathcal{T}}^{j,2}$, $\forall n$ in the cache memory of user $k$ if $k \in \mathcal{N}(\Gamma_j)$ and $Index(j,k)\in \mathcal{T}$. 

In addition, the server generates $h{{\hat K}\choose{t_2+1}}$ independent keys. Each key is uniformly distributed with length $\frac{\frac{F}{r}-\frac{NF}{D+\hat{K}-t_1}}{{{\hat{K}}\choose {t_2}}}F$ bits. We denote each key by $K_{\mathcal{T}_K}^u$, where $u\!=\!1,..,h$, and $\mathcal{T}_K\subseteq [\hat{K}], |\mathcal{T}_K|=t_2+1$. User $i$ stores the keys $K_{\mathcal{T}_K}^u$, $\forall u\! \in\! \mathcal{N}(U_i)$, whenever $Index(u,i)\in \mathcal{T}_K$. Also, the sever generates the random keys $K_j^i$ each of length $\frac{NF(\hat{K}-t_1)}{(D+\hat{K}-t_1)\hat{K}}$ bits, for $i=1,..,h$ and $j=1,..,\hat{K}$. $K_j^i$ will be cached by relay $i$ and user $k$ with $Index(i,k)=j$. Therefore, the cache contents at the relay nodes and end users are given by 
 \begin{equation}
V_k\!= \left\{f_{n}^{k,1}, K_j^k:  \ \forall n, j \right\},
 \end{equation}
\begin{align}
Z_k=&\left\{f_{n,k}^{j,1}, K_l^j,f_{n,\mathcal{T}}^{j,2},K_{\mathcal{T}_K}^j\!:\forall n, \forall j \!\in\! \mathcal{N}(U_k), Index(j,k)\!\in\! \mathcal{T}, \mathcal{T}_K, Index(j,k)\!=\!l \right\}.
\end{align}

The accumulated number of bits cached by each relay is given by 
\begin{align}\label{memacc11r}
D\frac{NF}{D+\hat{K}-t_1}+\hat{K}\frac{NF(\hat{K}-t_1)}{(D+\hat{K}-t_1)\hat{K}} =\frac{DNF+NF\hat{K}-NFt_1}{D+\hat{K}-t_1}=NF.
\end{align}
The accumulated of bits at each end user is given by 
\begin{align}\label{memacc11}
&\frac{Dr{{\hat K-1}\choose{t_2-1}}}{{{\hat K}\choose{t_2}}}|f_{n}^{i,2}|+\frac{r{{\hat K-1}\choose{t_2}}F}{{{\hat K}\choose{t_2}}}|f_{n}^{i,2}|+\frac{Drt_1}{\hat{K}}|f_{n}^{i,1}|+\frac{r(\hat{K}-t_1)}{\hat{K}}|f_{n}^{i,1}|\nonumber \\ &\quad=\frac{Drt_2}{\hat{K}}|f_{n}^{i,2}|+\frac{r(\hat{K}-t_2)}{\hat{K}}|f_{n}^{i,2}|+\frac{Drt_1}{\hat{K}}|f_{n}^{i,1}|+\frac{r(\hat{K}-t_1)}{\hat{K}}|f_{n}^{i,1}|
\nonumber \\ 
&\quad= F+\frac{t_2(D-1)F}{\hat{K}}+\frac{(t_1-t_2)r(D-1)NF}{\hat{K}(N+\hat{K}-t_1)}=MF,
\end{align}
thus satisfying the memory constraints on all caches. 
\subsection{Coded Delivery Phase}
At the beginning of the delivery phase, the demand vector $\bm d$ is announced in the network. For each relay node $\Gamma_j$, at each transmission instance, we consider $\mathcal{S}\subseteq[\hat{K}]$, where $|\mathcal{S}|=t_2+1$. For each $\mathcal{S}$, the server sends to the relay node $\Gamma_j$, the signal 
\begin{equation}
X_{j,\bm d}^{\mathcal{S}}=K^j_{\mathcal{S}} \bigoplus_{\{i:i \in \mathcal{N}(\Gamma_j),\ Index(j,i)\in \mathcal{S}\}} f^{j}_{d_i,\mathcal{S} \setminus\{Index(j,i)\}}.
\end{equation}
In total, the server transmits to $\Gamma_j$, the signal $X_{j,\bm d}=\bigcup_{\mathcal{S}\subseteq[\hat K]:|\mathcal{S}|=t_2+1} \{X_{j,\bm d}^{\mathcal{S}}\}.$
Then, $\Gamma_j$ forwards the signal $X_{j,\bm d}^{\mathcal{S}}$ to user $i$ whenever $Index(j,i)\in\mathcal{S}$. In addition, the relay $\Gamma_j$ sends $f_{d_i}^{j,1} \setminus f_{d_i,i}^{j,1}$ to user $i$ encrypted by the key $K_k^j$ such that $Index(j,i)=k$, i.e., we have  
\begin{align}
Y_{j,\bm d,i}=\Big\{ K_k^j\oplus\{f_{d_i}^{j,1} \setminus f_{d_i,i}^{j,1}\}\Big\} \bigcup_{\mathcal{S}\subseteq[\hat K]:|\mathcal{S}|=t_2+1,Index(j,i)\in\mathcal{S}} \big\{X_{j,\bm d}^{\mathcal{S}}\big\}.
\end{align} 
First, user $i$ can decrypt its received signals using the cached keys. Then, it can recover the pieces $\left\{f^{j,2}_{d_i,\mathcal{T}}: \mathcal{T}\subseteq [\hat{K}]\setminus \{Index(j,i)\} \right\}$ from the signals received from $\Gamma_j$, utilizing its cache's contents. With its cached contents, user $i$ can recover $f^{j,2}_{d_i}$. In addition, user $i$ directly gets $f_{d_i}^{j,1}$ from its the signal transmitted by relay $j$. Thus, it can obtain $f^{j}_{d_i}$. Since, user $i$ receives signals from $r$ different relay nodes, it can obtain the encoded symbols $f^{j}_{d_i}$, $\forall j \in \mathcal{N}(U_i)$, and is able to successfully reconstruct its requested file $W_{d_i}$. 
\begin{remark}\label{secremark}
In total, the server sends $h{{\hat{K}}\choose {t_2+1}}$ signals, each of which is encrypted using a one-time pad that has length equal to the length of each subfile ensuring prefect secrecy \cite{shannon1949}. Observing any of the transmitted signals without knowing the encryption key will not reveal any information about the database files \cite{shannon1949}. The same applies for the messages transmitted by the relays. Thus, (\ref{secureconst}) is satisfied.
\end{remark}
\subsection{Secure Delivery Rates} 
Denote the secure delivery rates in the first and second hop with $R_1^s$ and $R_2^s$, respectively. Each relay node is responsible for ${{\hat{K}}\choose {t_2+1}}$ transmissions, each of length $\frac{|f_{n}^{i,2}|}{{{\hat{K}}\choose {t_2}}}$, thus the transmission rate in bits from the server to each relay node is 
\begin{align}
R_1^sF=\frac{{{\hat{K}}\choose {t_2+1}}}{{{\hat{K}}\choose {t_2}}}|f_{n}^{i,2}|\!=\!\frac{\hat{K}-t_2}{(t_2+1)}\left(\frac{F}{r}-\frac{NF}{D+\hat{K}-t_1}\right).
\end{align} 
 $\Gamma_j$ forwards ${{\hat{K}-1}\choose {t_2}}$ from its received signals to each connected end users. In addition, it transmits a message of size $\frac{NF(\hat{K}-t_1)}{(D+\hat{K}-t_1)\hat{K}}$ bits from its cached contents to each user, thus we have 
\begin{align}
R_2^sF&=\frac{{{\hat{K}-1}\choose {t_2}}}{{{\hat{K}}\choose {t_2}}}|f_{n}^{i,2}|+\frac{NF(\hat{K}-t_1)}{(D+\hat{K}-t_1)\hat{K}}=\left(1-\frac{t_2}{\hat{K}}\right)\left(\frac{F}{r}-\frac{NF}{D+\hat{K}-t_1}\right)+\frac{NF(\hat{K}-t_1)}{(D+\hat{K}-t_1)\hat{K}}\nonumber\\ 
&=\frac{F}{r}\left(1-\frac{M-1}{N-1}\right).
\end{align} 
We thus obtain the following theorem.
\begin{theorem}
The normalized transmission rates with secure delivery, for $N\geq0$, $M=1+\frac{t_2(D-1)}{\hat{K}}+\frac{(t_1-t_2)r(D-1)N}{\hat{K}(N+\hat{K}-t_1)}$, $t_1,t_2\!\in\!\{0,1,..,\hat{K}\}$ and $\frac{t_1}{\hat{K}}\leq\frac{N}{D+\hat{K}-t_1}$, are upper bounded by 
\begin{align}
R_1^s\leq\!\frac{\hat{K}-t_2}{r(t_2+1)}\left(1-\frac{Nr}{D+\hat{K}-t_1}\right), \qquad R_2^s\leq\frac{1}{r}\left(1-\frac{M-1}{D-1}\right).
\end{align}
In addition, the convex envelope of these points is achievable by memory sharing.  
\end{theorem}
For the special case of no caches at the relays, i.e., $N=0$, $t_1=0$, we obtain the following upper bound on the secure delivery rates.  
\begin{corollary}
The normalized transmission rates with secure delivery, for $N=0$, $M\!=\!1\!+\!\frac{t(D\!-\!1)}{\hat{K}}$, and $t\!\in\!\{0,1,..,\hat{K}\}$, are upper bounded by 
\begin{align}
R_1^s\leq\!\frac{\hat{K}\left(1-\frac{M-1}{D-1}\right)}{r\left(\hat{K}\frac{M-1}{D-1}+1\right)}, \qquad R_2^s\leq\frac{1}{r}\left(1-\frac{M-1}{D-1}\right).
\end{align}
In addition, the convex envelope of these points is achievable by memory sharing.  
\end{corollary}

\begin{remark}\label{secremarkzerorates}
Under secure delivery, we place keys in the memories of both end user and relays, i.e., we divide the cache between storing data and keys. Observe that the rate of the second hop is the complement of the \textit{data} caching gain of end user and is determined by $M$ only. In addition, whenever $M\geq D$, each user can cache the entire library and there is no need for caching keys as $R_1^s=R_2^s=0$.
On the other hand, the rate of the first hop $R_1^s$ is affected by both $M$ and $N$. We achieve zero rate over the first hop whenever $t_2=\hat{K}$, i.e., we need $M\geq D-\frac{(D-1)Nr}{\hat{K}+N}$. 
\end{remark}

\section{Combination networks with secure caching}\label{sec:achconf}
Next, we consider \textit{secure caching}, i.e., an end user must be able to recover its requested file, and must \textit{not} be able to obtain any information about the remaining files, i.e., for $\delta\!>\!0$
\begin{equation}\label{secrtiveconst}
\max_{\bm d,\mathcal{V}} I(\bm W_{-d_k}; \{Y_{j,\bm d,k}:j\in\mathcal{N}(U_k)\},Z_k)< \delta, 
\end{equation}
where $\bm W_{-d_k}\!=\!\{W_1,..,W_N\}\backslash \{W_{d_k}\}$, i.e., the set of all files except the one requested by user $k$.

In our achievability, we utilize secret sharing schemes \cite{secretsharing} to ensure that no user is able to obtain information about the files from its cached contents. The basic idea of the secret sharing schemes is to encode the secret in such a way that accessing a subset of shares does not suffice to reduce the uncertainty about the secret. For instance, if the secret is encoded into the scaling coefficient of a line equation, the knowledge of one point on the line does not reveal any information about the secret as there remain infinite number of possibilities to describe the line. One can learn the secret only if two points on the line are provided. 

In particular, we use a class of secret sharing scheme known as \textit{non-perfect secret sharing schemes}, defined as follows. 
\begin{Definition}
\cite{secretsharing} \cite{blakley1984security} For a secret $W$ with size $F$ bits, an $(m,n)$ \textit{non-perfect secret sharing scheme} generates $n$ shares, $S_1,S_2,..S_n$, such that accessing any $m$ shares does not reveal any information about the file $W$, i.e., 
\begin{equation}
I(W;\mathcal{S})=0, \quad \forall \mathcal{S}\subseteq\{S_1,S_2,..S_n\}, |\mathcal{S}|\leq m.
\end{equation}
Furthermore, $W$ can be losslessly reconstructed from the $n$ shares, i.e.,
\begin{equation}
H(W|S_1,S_2,..,S_n)\!\!=\!0.
\end{equation}  
\end{Definition}
For large enough $F$, an $(m,n)$ secret sharing scheme exists with shares of size equal to $\frac{F}{n-m}$ bits \cite{blakley1984security,secretsharing}.
 
\subsection{Cache Placement Phase}\label{conf_cache}
Again, as a first step, the server divides each file into $r$ equal-size subfiles. Then, it encodes them using an $(h,r)$ maximum distance separable (MDS) code. We denote by $f_n^i$ the resulting encoded symbol, where $n$ is the file index and $i=1,2,..,h$. 
For $M=\frac{tD}{\hat{K}-t}(1-\frac{Nr}{D})$ and $t\!\in\!\{0,1,..,\hat{K}\!-\!1\}$, we divide each encoded symbol into two parts, $f_n^{i,1}$ with size $\frac{NF}{D}$ bits and $f_n^{i,2}$ with size $\frac{F}{r}-\frac{NF}{D}$ bits. The parts $\{f_n^{i,1}: \ \forall n\}$ will be cached in the memory of relay $\Gamma_i$ and will not be cached by any user.

 Each of the symbols $f_n^{i,2}$ is encoded using a $\left({{\hat{K}-1}\choose {t-1}},\!{{\hat{K}}\choose {t}}\!\right)$ secret sharing scheme from \cite{blakley1984security,secretsharing}. The resulting shares are denoted by $S_{n,\mathcal{T}}^j$, where $n$ is the file index i.e., $n\in\{1,..,N\}$, $j$ is the index of the encoded symbol, i.e., $j=1,..,h$, and $\mathcal{T}\subseteq [\hat{K}], |\mathcal{T}|=t$. Each share has size
\begin{equation}
F_s=\frac{\frac{F}{r}-\frac{NF}{D}}{{{\hat{K}}\choose {t}}-{{\hat{K}-1}\choose {t-1}}}=\frac{t\left(1-\frac{Nr}{D}\right)}{r(\hat{K}-t){{\hat{K}-1}\choose {t-1}}}F \mbox{ bits}.
\end{equation} 
The server allocates the shares $S_{n,\mathcal{T}}^j$, $\forall n$ in the cache of user $k$ whenever $j \in \mathcal{N}(U_k)$ and $Index(j,k) \in \mathcal{T}$. Therefore, at the end of cache placement phase, the contents of the cache memory at relay $j$ and user $k$ are given by 
 \begin{equation}
V_j\!=\!\left\{f_{n}^{j,1}:  \ \forall n \right\},  \qquad Z_k=\left\{S_{n,\mathcal{T}}^j: k\!\in\! \mathcal{N}(\Gamma_j),\  Index(j,k)\in \mathcal{T}, \ \forall n\right\}.
 \end{equation}

\begin{remark}
Each user stores $Dr{{\hat K-1}\choose{t-1}}$ shares, thus the accumulated number of bits stored in each cache memory is 
\begin{align}\label{memacc2}
Dr&{{\hat K-1}\choose{t-1}}\frac{t\left(1-\frac{Nr}{D}\right)}{r(\hat{K}-t){{\hat{K}-1}\choose {t-1}}}F=\frac{tD}{\hat{K}-t}\left(1-\frac{Nr}{D}\right)F=MF.
\end{align}
Clearly, the proposed scheme satisfies the cache capacity constraint at both relays and end users. Furthermore, from (\ref{memacc2}), we can get $t=\frac{\hat{K}M}{D+M-Nr}$.
\end{remark}
\subsection{Coded Delivery Phase}
At the beginning of the delivery phase, each user requests a file from the server. First, we focus on the transmissions from the server to $\Gamma_j$. At each transmission instance, we consider $\mathcal{S}\subseteq[\hat{K}]$, where $|\mathcal{S}|=t+1$. For each $\mathcal{S}$, the server transmits the following signal to $\Gamma_j$  
\begin{equation}\label{rece_conf1}
X_{j,\bm d}^{\mathcal{S}}=\bigoplus_{\{i:i \in \mathcal{N}(\Gamma_j),\ Index(j,i)\in \mathcal{S}\}} S^{j}_{d_i,\mathcal{S} \setminus\{Index(j,i)\}}.
\end{equation}
In total, the server transmits to $\Gamma_j$, the signal $X_{j,\bm d}=\bigcup_{\mathcal{S}\subseteq[\hat K]:|\mathcal{S}|=t+1} \{X_{j,\bm d}^{\mathcal{S}}\}.$
 Then, $\Gamma_j$ forwards the signal $X_{j,\bm d}^{\mathcal{S}}$ to user $i$ whenever $Index(j,i)\in\mathcal{S}$. In addition, $\Gamma_j$ sends directly $f_{d_i}^{j,1}$ to user $i$. Therefore, we have 
 \begin{align}
Y_{j,\bm d,i}=\{ f_{d_i}^{j,1} \} \bigcup_{\mathcal{S}\subseteq[\hat K]:|\mathcal{S}|=t+1,Index(j,i)\in\mathcal{S}} \big\{X_{j,\bm d}^{\mathcal{S}}\big\}.
\end{align}
 
  User $i$ can recover $\{S^{j}_{d_i,\mathcal{T}}: \mathcal{T}\subseteq [\hat{K}]\setminus \{Index(j,i)\}, |\mathcal{T}|=t \}$ from the signals received from $\Gamma_j$, utilizing its cache's contents. Adding these shares to the ones in its cache, i.e., $S^{j}_{d_i,\mathcal{T}}$ with $Index(j,i)\in\mathcal{T}$, user $i$ can decode the encoded symbol $f^{j,2}_{d_i}$ from its ${{\hat{K}}\choose{t}}$ shares.
Since, user $i$ receives signals from $r$ different relay nodes, it obtains the encoded symbols $f^{j}_{d_i}$, $\forall j \in \mathcal{N}(U_i)$, and can reconstruct $W_{d_i}$.
\subsection{Secure Caching Rates}
Under secure caching requirement, we denote the first and second hop rates as $R_1^c$ and $R_2^c$, respectively. Since, each relay node is responsible for ${{\hat{K}}\choose {t+1}}$ transmissions, each of length $F_s$, the transmission rate, in bits, from the server to each relay node is 
\begin{align}
R_1^cF\!&=\frac{t{{\hat{K}}\choose {t+1}}\left(1-\frac{Nr}{D}\right)F}{r(\hat{K}-t){{\hat{K}-1}\choose {t-1}}} = \frac{\hat{K}\left(1-\frac{Nr}{D}\right)F}{r(t\!+\!1)}=\frac{\hat{K}(D\!+\!M-rN)}{r\left(\!(\hat{K}\!+\!1)M\!+\!D\!-rN\right)}\left(1-\frac{Nr}{D}\right)F.
\end{align} 
Then, each relay forwards ${{\hat{K}-1}\choose {t}}$ from these signals to each of its connected end users. In addition, each relay forwards $\frac{NF}{D}$ bits from its cache to each of these users, therefore
\begin{align}
R_2^cF={{\hat{K}-1}\choose {t}}\frac{t\left(1-\frac{Nr}{D}\right)}{r(\hat{K}-t){{\hat{K}-1}\choose {t-1}}}F+\frac{NF}{D}=\frac{1}{r}F.
\end{align} 
Consequently, we have the following theorem.
\begin{theorem}\label{securecach}
The normalized rates with secure caching, for $0\leq N\leq \frac{D}{r}$, $M\!\!=\!\!\frac{tD}{\hat{K}-t}(1-\frac{Nr}{D})$, and $t\!\in\!\{0,1,..,\hat{K}\!-\!1\}$, are upper bounded by
\begin{align}\label{securecacherate1}
R_1^c\leq \frac{\hat{K}(D\!+\!M-rN)}{r\left(\!(\hat{K}\!+\!1)M\!+\!D\!-rN\right)}\left(1-\frac{Nr}{D}\right), \qquad R_2^c\leq\frac{1}{r}.
\end{align}
The convex envelope of these points is achievable by memory sharing.  
\end{theorem}
\begin{remark}
Secret sharing encoding guarantees that no user is able to reconstruct any file from its cache contents only, as the cached shares are not sufficient to reveal any information about any file. In addition, the only new information in the received signals by any end user is the shares related to its requested file. Thus, (\ref{secrtiveconst}) is satisfied.
\end{remark}
For the special case of no relay caches, we obtain the following corollary.
\begin{corollary}\label{securecachno}
The normalized rates with secure caching, for $N=0$, $M=\frac{tD}{\hat{K}-t}$, and $t\!\in\!\{0,1,..,\hat{K}\!-\!1\}$, are upper bounded by
\begin{align}
R_1^c\leq \frac{\hat{K}(D\!+\!M)}{r\left(\!(\hat{K}\!+\!1)M\!+\!D\!\right)}, \qquad R_2^c\leq\frac{1}{r}.
\end{align}
The convex envelope of these points is achievable by memory sharing.  
\end{corollary}
\begin{remark}
From the results in \cite{ravindrakumar2016fundamental} and \cite{zewail2016fundamental}, it is known that, under secure caching requirement, the number bits received by any user is lower bounded by the file size. Thus, $R_2^c$ in (\ref{securecacherate1}) is optimal. In addition, to achieve zero rate over the first hop, we need $N\geq \frac{D}{r}$, where we distribute the library over the relays caches and we do not need to utilize the caches at the end users due to the unicast nature of the network links. 
\end{remark}

\section{Combination networks with secure caching and secure delivery}\label{sec:achconf_sec}
Now, we investigate the network under the requirements studied in Sections \ref{sec:achdelivery} and \ref{sec:achconf}, simultaneously. The achievability scheme utilizes both one-time pads and secret sharing.
\subsection{Cache Placement Phase}
For $M\!\!= 1+\frac{tD}{\hat{K}-t}(1-\frac{rN}{D+\hat K})$, and $t\in\{0,1,..,\hat{K}\!-\!1\}$, after encoding each file using an $(h,r)$ MDS code, we divide each encoded symbol into two parts, $f_n^{i,1}$ with size $\frac{NF}{D+\hat K}$ bits and $f_n^{i,2}$ with size $\frac{F}{r}-\frac{NF}{D+\hat K}$ bits. Only $\Gamma_i$ caches the parts $\{f_n^{i,1}: \ \forall n\}$.

 Each of the symbols $f_n^{i,2}$ is encoded using a $\left({{\hat{K}-1}\choose {t-1}},\!{{\hat{K}}\choose {t}}\!\right)$ secret sharing scheme from \cite{blakley1984security,secretsharing}. The resulting shares are denoted by $S_{n,\mathcal{T}}^j$, where $n$ is the file index i.e., $n\in\{1,..,N\}$, $j$ is the index of the encoded symbol, i.e., $j=1,..,h$, and $\mathcal{T}\subseteq [\hat{K}], |\mathcal{T}|=t$. Each share has size
\begin{equation}
F_s=\frac{\frac{F}{r}-\frac{NF}{D+\hat{K}}}{{{\hat{K}}\choose {t}}-{{\hat{K}-1}\choose {t-1}}}=\frac{t\left(1-\frac{Nr}{D+\hat{K}}\right)}{r(\hat{K}-t){{\hat{K}-1}\choose {t-1}}}F \mbox{ bits}.
\end{equation} 
The server allocates the shares $S_{n,\mathcal{T}}^j$, $\forall n$ in the cache of user $k$ whenever $j \in \mathcal{N}(U_k)$ and $Index(j,k) \in \mathcal{T}$.

 Furthermore, the server generates $h{{\hat K}\choose{t+1}}$ independent keys. Each key is uniformly distributed with length $F_s$ bits. We denote each key by $K_{\mathcal{T}_K}^u$, where $u\!=\!1,..,h$, and $\mathcal{T}_K\subseteq [\hat{K}], |\mathcal{T}_K|=t+1$. User $i$ stores the keys $K_{\mathcal{T}_K}^u$, $\forall u\! \in\! \mathcal{N}(U_i)$, whenever $Index(u,i)\in \mathcal{T}_K$. Also, the sever generates the random keys $K_j^i$ each of length $\frac{NF}{D+\hat K}$ bits, for $i=1,..,h$ and $j=1,..,\hat{K}$, which will be cached by relay $i$ and user $k$ with $Index(i,k)=j$.

 Therefore, at the end of cache placement phase, the contents of the cache memory at relay $j$ and user $k$ are given by 
 \begin{equation}
 V_j\!=\!\left\{f_{n}^{j,1}, K_u^j:  \ \forall n,u \right\},
 \end{equation}
\begin{align}
Z_k=&\left\{S_{n,\mathcal{T}}^j,K_{\mathcal{T}_K}^j\!, K_l^j: \forall n, \forall j\!\in\! \mathcal{N}(U_k), Index(j,k)\!\in\! \mathcal{T}, \mathcal{T}_K, Index(j,k)=l \right\}.
\end{align}

\begin{remark}
In addition to the keys, each user stores $Dr{{\hat K-1}\choose{t-1}}$ shares, thus the accumulated number of bits stored in each cache memory is 
\begin{align}\label{memacc3}
Dr&{{\hat K-1}\choose{t-1}}\frac{t\left(1-\frac{Nr}{D+\hat{K}}\right)}{r(\hat{K}-t){{\hat{K}-1}\choose {t-1}}}F+r{{\hat K-1}\choose{t}}\frac{t\left(1-\frac{Nr}{D+\hat{K}}\right)}{r(\hat{K}-t){{\hat{K}-1}\choose {t-1}}}F+r\frac{NF}{D+\hat K}\nonumber \\
&=\frac{Dt\left(1-\frac{Nr}{D+\hat{K}}\right)}{r(\hat{K}-t)}F+\left(1-\frac{Nr}{D+\hat{K}}\right)F+r\frac{NF}{D+\hat K}=MF. 
\end{align}
Thus, the scheme satisfies the memory constraints, and we get $ t=\frac{\hat{K}(M-1)(D+\hat{K})}{(D+\hat K)(M+D-1)+rND}$. 
\end{remark}
\subsection{Coded Delivery Phase}
The delivery phase begins with announcing the demand vector to all network nodes. For $\Gamma_j$, at each transmission instance, we consider a $\mathcal{S}\subseteq[\hat{K}]$, where $|\mathcal{S}|=t+1$. For each $\mathcal{S}$, the server transmits to $\Gamma_j$, the following signal  
\begin{equation}\label{rece_conf3}
X_{j,\bm d}^{\mathcal{S}}=K^j_{\mathcal{S}}\bigoplus_{\{i:i \in \mathcal{N}(\Gamma_j),\ Index(j,i)\in \mathcal{S}\}} S^{j}_{d_i,\mathcal{S} \setminus\{Index(j,i)\}},
\end{equation}
i.e., the server transmits to $\Gamma_j$, the signal $X_{j,\bm d}=\bigcup_{\mathcal{S}\subseteq[\hat K]:|\mathcal{S}|=t+1} \{X_{j,\bm d}^{\mathcal{S}}\}$. 
Then, $\Gamma_j$ forwards the signal $X_{j,\bm d}^{\mathcal{S}}$ to user $i$ whenever $Index(j,i)\in\mathcal{S}$. In addition, $\Gamma_j$ sends $f_{d_i}^{j,1}$ encrypted by $K_u^j$ to user $i$ such that $Index(j,i)=k$. After decrypting the received signals, user $i$ get $f_{d_i}^{j,1}$ and can extract the set of shares $\{S^{j}_{d_i,\mathcal{T}}: \mathcal{T}\subseteq [\hat{K}]\setminus \{Index(j,i)\}, |\mathcal{T}|=t \}$ from the signals received from $\Gamma_j$. These shares in addition to the ones in its cache, i.e., $S^{j}_{d_i,\mathcal{T}}$ with $Index(j,i)\in\mathcal{T}$, allow user $i$ to decode $f^{j,2}_{d_i}$ from its ${{\hat{K}}\choose{t}}$ shares. Since, user $i$ receives signals from $r$ different relay nodes, it obtains $\{f^{j}_{d_i}$, $\forall j \in \mathcal{N}(U_i)\}$, then decodes $W_{d_i}$.
 \subsection{Secure Caching and Secure Delivery Rates}
 We refer to the first and second hop rates as $R_1^{sc}$ and $R_2^{sc}$, respectively. Each relay node sends ${{\hat{K}}\choose {t+1}}$ signals, each of length $F_s$, thus we have
\begin{align}
R_1^{sc}F&={{\hat{K}}\choose {t+1}}\frac{t\left(1-\frac{Nr}{D+\hat{K}}\right)}{r(\hat{K}-t){{\hat{K}-1}\choose {t-1}}}F=\frac{\hat{K}}{r(t+1)}\left(1-\frac{Nr}{D+\hat{K}}\right)F\nonumber \\ 
&=\frac{\hat{K}\left(rND+(D+\hat{K})(M+D-1)\right)}{r\left(rND+(D+\hat{K})[D+(M-1)(\hat{K}+1)]\right)}\left(1-\frac{Nr}{D+\hat{K}}\right)F.
\end{align} 
In the second hop, each relay node is responsible for forwarding ${{\hat{K}-1}\choose {t}}$ from its received signals to each of its connected end users, in addition, it transmits $\frac{NF}{D+\hat K}$ bits from its cache, thus 
\begin{align}
R_2^{sc}F={{\hat{K}-1}\choose {t}}\frac{t\left(1-\frac{Nr}{D+\hat{K}}\right)}{r(\hat{K}-t){{\hat{K}-1}\choose {t-1}}}F+\frac{NF}{D+\hat K}=\frac{1}{r}F.
\end{align} 
Therefore, we can obtain the following theorem.
\begin{theorem}
Under secure delivery and secure caching requirements, for $0\leq N\leq \frac{D+\hat{K}}{r}$, $M\!\!= 1+\frac{tD}{\hat{K}-t}(1-\frac{rN}{D+\hat K})$, and $t\in\{0,1,..,\hat{K}\!\!-\!\!1\}$, the transmission rates are upper bounded by
\begin{align}
R_1^{sc}\leq \frac{\hat{K}\left(rND+(D+\hat{K})(M+D-1)\right)}{r\left(rND+(D+\hat{K})[D+(M-1)(\hat{K}+1)]\right)}\left(1-\frac{Nr}{D+\hat{K}}\right), \qquad R_2^{sc}\leq\frac{1}{r}.
\end{align}
In addition, the convex envelope of these points is achievable by memory sharing.  
\end{theorem}
For the case where there is no caches at the relays, we have
\begin{corollary}
Under secure delivery and secure caching requirements, for $N=0$, $M\!\!=\!\!\frac{tD}{\hat{K}-t}\!+\!1$,  and $t\in\{0,1,..,\hat{K}\!\!-\!\!1\}$, the transmission rates are upper bounded by
\begin{align}
R_1^{sc}\leq\frac{\hat{K}(D+M-1)}{r\left((\hat{K}+1)(M-1)+D\right)}, \qquad R_2^{sc}\leq\frac{1}{r}.
\end{align}
In addition, the convex envelope of these points is achievable by memory sharing.  
\end{corollary}
\vspace{-.1 in}
\section{Numerical Results and Discussion}\label{sec:discuss}
In this section, we discuss the insights gained from our study and demonstrate the performance of our proposed techniques. We focus on the achievable rates over the links of each hop of communication.
\vspace{-.1 in}   
\subsection{Achievable Rates over the First Hop}
\begin{figure}[t]
\includegraphics[width=0.70\textwidth ,height=2.9 in]{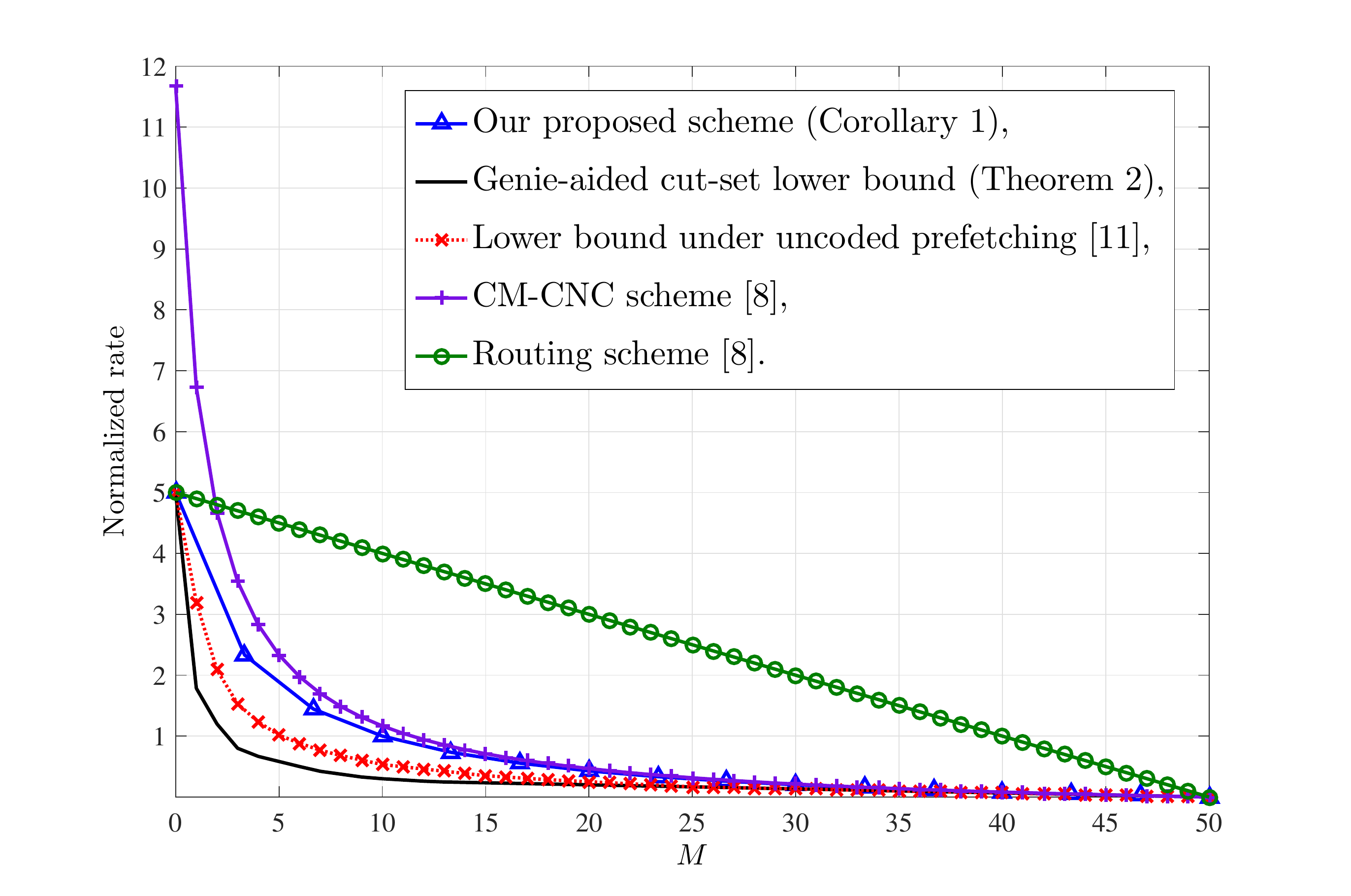}
\centering
\caption{\small Lower and upper bounds for $K=35$, $N=0$, $D\!=\!50$, $h\!=\!7$ and $r\!=\!3$ .}\label{comp_h7r3}
\vspace{-.2 in}
\end{figure}

Fig. \ref{comp_h7r3} shows the comparison between the achievable normalized rate of our proposed scheme in Corollary 1 (the special case with no caching at the relays), lower bound in Theorem \ref{thm:lower}, lower bound under uncoded prefetching \cite{wan2017combination}, the coded multicasting and combination network coding (CM-CNC) scheme \cite{ji2015caching}, and the routing scheme from \cite{ji2015caching}. We can see that our proposed scheme outperforms the ones in \cite{ji2015caching}. We remark that in this special case of no caching relays, the lower bounds in subsection \ref{sec:lower_R1} reduce to the ones in \cite{ji2015caching}. Therefore, the same order optimality as in \cite[Theorem 4]{ji2015caching} applies.

In Fig. \ref{h7_r4_M1}, we plot the normalized rate for different relay cache sizes. It can be observed that the normalized rates are decreasing functions of the memory capacities and whenever $M+rN\geq 60$, $R_1=0$, while $R_2=0$ if $M\geq 60$. This shows how the cache memories at the relay nodes as well as the ones at the end users can completely replace the main server during the delivery phase. We note that the gap between lower and upper bound decreases as $M$ increases.
 \begin{figure}[t]
 \includegraphics[width=0.70\textwidth ,height=2.9 in]{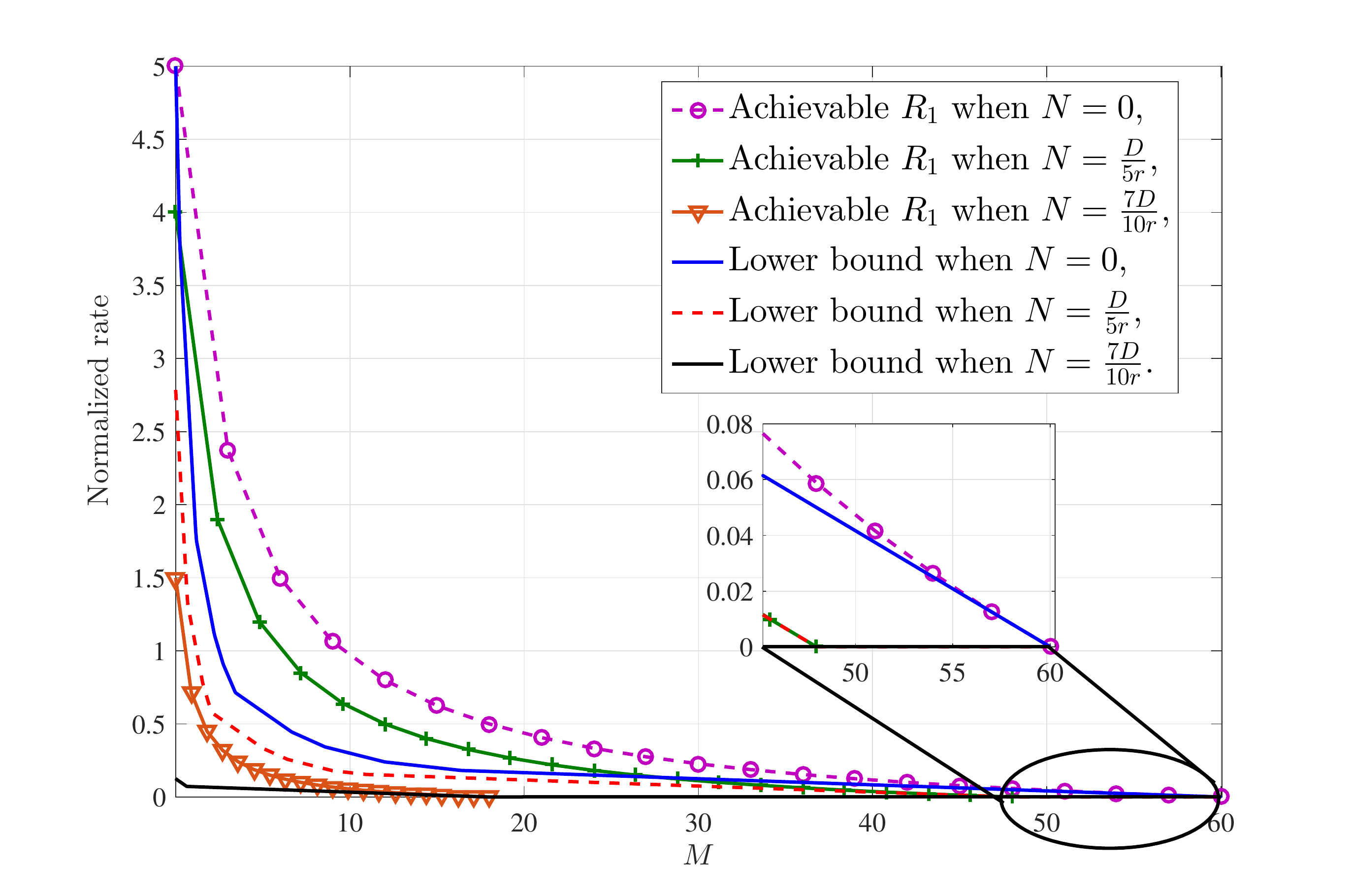}
\centering
\caption{\small Lower and upper bounds for $K=35$, $D\!=\!60$, $h\!=\!7$ and $r\!=\!4$.}\label{h7_r4_M1}
\vspace{-.1 in}
\end{figure}

\vspace{-.2 in}
\subsection{Optimality over the Second Hop}
From the cut set bound in Theorem \ref{thm:lower}, we see that our achievable rate over the second hop is optimal. Thus, the total delivery load per relay is minimized.

For the network with no caches at the relays, it has been shown in \cite{wan2017combination} that the proposed schemes achieve lower rates over the first hop compared with the scheme in \cite{tang2016coded}, i.e., achieves lower $R_1$ than our scheme in Section \ref{sec:nonsec}, for the case where $M=N/K$. Additionally, it has been shown that the schemes in \cite{wan2017combination} achieve the optimal rate under uncoded prefetching for $r=h-1$. Note that the scheme based on interference alignment in \cite{wan2017combination} for the case where $M=D/K$ achieves lower rates over the first hop, however the achievable rate during the second hop is not optimal. As an example consider the network with $D=K=6$, $h=4$, $r=2$ and $M=1$, the normalized optimal delivery load during the second hop is $\frac{5}{12}$ and it is achievable by our scheme. The scheme in \cite[Section IV-B]{wan2017combination} achieves normalized delivery load of $\frac{7}{12}$. On the other hand, in this example, the scheme in \cite[Section IV-B]{wan2017combination} achieves $R_1=\frac{2}{3}$, while our scheme achieves $1$. Therefore, the total normalized network load, i.e., $hR_1\!+\!rKR_2$, under the scheme in \cite[Section IV-B]{wan2017combination} is $\frac{29}{3}$, while our scheme achieves $9$. This example demonstrates to the importance of ensuring the optimality over the second hop in order to reduce the overall network load.   
\vspace{-.19 in}
\subsection{Performance with Secrecy Requirements}

In Figs. \ref{fig_diff_requirementsR1} and \ref{fig_diff_requirementsR2}, we compare the achievable rates under different secrecy scenarios. From these figures, we observe that the cost of imposing secure delivery is \textit{negligible} for realistic system parameters. The gap between the achievable rates of the system without secrecy and the system with secure delivery vanishes as $M$ increases. Same observation holds for the gap between the rates with secure caching and those with secure caching and secure delivery.  

In addition, achievable rates over the second hop is optimal, i.e., achieves the cut set bound. In particular, under secure delivery, each user caches a fraction $\frac{M-1}{N-1}$ of each file, and the total data received by any end user under secure delivery equals $(1\!-\!\frac{M-1}{N-1})F$, which is the minimum number of bits required to reconstruct the requested file. Similarly, in the two remaining scenarios, we know from the result in reference \cite{ravindrakumar2016fundamental} that the minimum number of bits required by each user to be able to recover its requested file is $F$, and our achievable schemes achieve this lower bound.  
\begin{figure}[t]
\includegraphics[width=0.7\textwidth ,height=2.8 in]{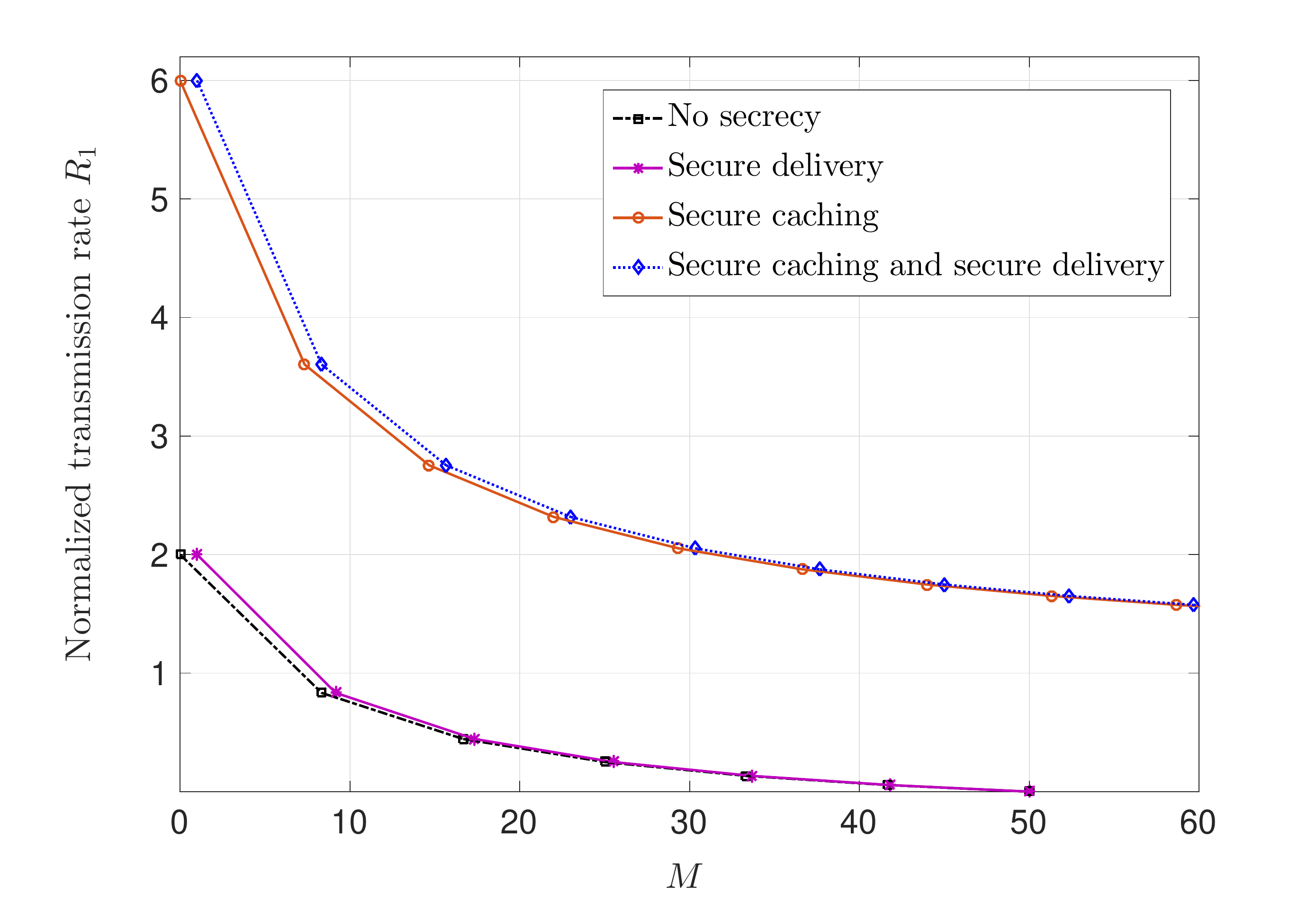}
\centering
\caption{\small Rates over the first hop under different system requirements for $N=0$, $D\!=\!50$, $K\!=\!15$, $h\!=\!5$ and $r\!=\!3$.}\label{fig_diff_requirementsR1}
\vspace{-.2 in}
\end{figure}
 \begin{figure}
\includegraphics[width=0.7\textwidth ,height=2.8 in]{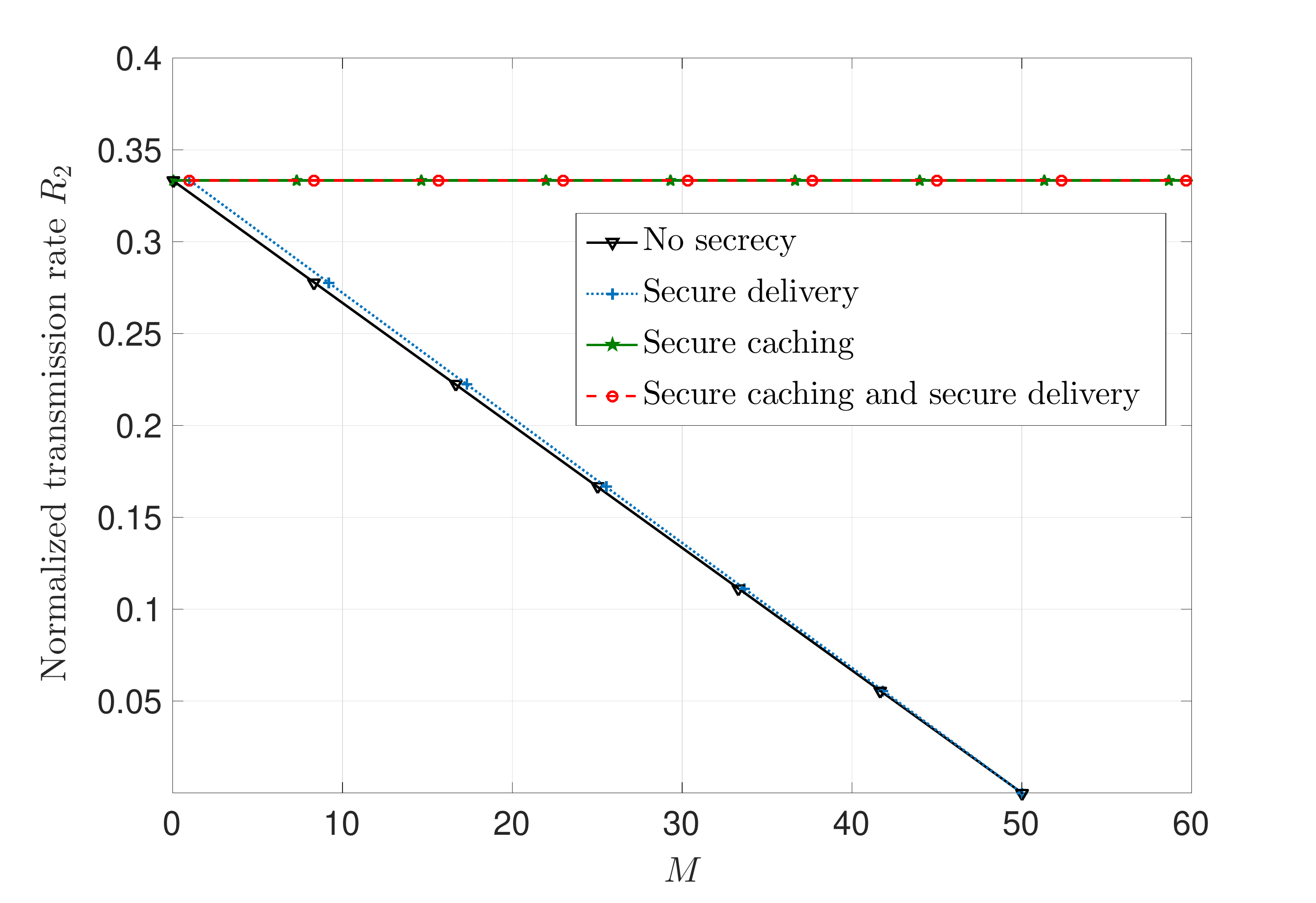}
\centering
\caption{\small Rates over the second hop under different system requirements for $N=0$, $D\!=\!50$, $K\!=\!15$, $h\!=\!5$ and $r\!=\!3$.}\label{fig_diff_requirementsR2}
\end{figure}
Another observation is that under secure caching requirement only (Section \ref{sec:achconf}), we do not need to use keys in order to ensure the secure caching requirement, in contrast with the general schemes in references \cite{ravindrakumar2016fundamental} and \cite{zewail2016fundamental}. This follows from the network structure, as the relay nodes unicast the signals to each of the end users. In particular, the received signals by user $k$ are formed by combinations of the shares in its memory and "fresh" shares of the requested file. Thus, at the end of communications, it has $r{{\hat{K}}\choose {t}}$ shares of the file $W_{d_k}$, and only $r{{\hat{K}-1}\choose {t-1}}$ shares of the remaining files, i.e., the secure caching requirement is satisfied, without the need to encrypt. In addition, for the case where $M=0$, i.e., no cache memory at the end users, secure caching is possible via routing, unlike the case in \cite{ravindrakumar2016fundamental}, where $M$ must be at least $1$. 

\vspace{-.15 in}
\begin{remark}
Corollaries 2-4 generalize our previous results that were limited to resolvable networks \cite{zewail2016coded}, i.e., we show the achievability of the rates in \cite{zewail2016coded} for any combination network. 
 \end{remark}

\vspace{-.4 in}
\section{Conclusion}\label{sec:con}
\vspace{-0.05 in}
In this work, we have investigated the fundamental limits two-hop cache-aided combination networks with caches at the relays and the end users, with and without security requirements. We have proposed a new coded caching scheme, by utilizing MDS coding and jointly optimizing the cache placement and delivery phases. We have shown that whenever the sum of the end user cache and the ones of its connected relays is sufficient to store the database, then there is no need for the server transmission over the first hop. We have developed genie-aided cut-set lower bounds on the rates and shown order optimality for the first hop and optimality for the second. 

We have next investigated combination networks with caching relays under secure delivery constraints, secure caching constraints, as well as both secure delivery and secure caching constraints. The achievability schemes, for each of these requirements, jointly optimize the cache placement and delivery phases, and utilize one-time padding and secret sharing. We have illustrated the impact of the network structure and relaying on the system performance after imposing different secrecy constraints.

 The decomposition philosophy using MDS codes we have utilized in this work allows adopting the ideas developed for the classical coded caching setup to cache-aided combination networks. Future directions in combination networks include caching with untrusted relays and considering the physical layer impairments in the delivery phase. 
\vspace{-.2 in}
 \bibliographystyle{IEEEtran}
\bibliography{IEEEabrv,cachingLib}
 \vspace{-.2 in} 
\end{document}